\documentclass[11pt]{article} 

\usepackage{fullpage}
\usepackage[utf8]{inputenc} 
\usepackage{amsmath,amsthm,amsfonts,amssymb}
\usepackage{color}
\usepackage[usenames,dvipsnames]{xcolor}
\usepackage{url}
\usepackage{float}
\usepackage{bbm}
\usepackage[backref=page]{hyperref}
\usepackage{wrapfig}

\hypersetup{
    bookmarksnumbered=true, 
    unicode=false, 
    pdfstartview={FitH}, 
    pdftitle={Dequantizing read-once quantum formulas}, 
    pdfauthor={Alessandro Cosentino, Robin Kothari, Adam Paetznick}, 
    pdfsubject={}, 
    pdfcreator={}, 
    pdfproducer={}, 
    pdfkeywords={}, 
    pdfnewwindow=true, 
    colorlinks=true, 
    linkcolor=blue, 
    citecolor=blue, 
    filecolor=blue, 
    urlcolor=blue 
}

\newtheorem{theorem}{Theorem}

\newtheorem{corollary}{Corollary}

\newtheorem{fact}{Fact}

\theoremstyle{definition}

\newtheorem{definition}{Definition}

\newcommand{\eqnref}[1]{\hyperref[#1]{{(\ref*{#1})}}}
\newcommand{\thmref}[1]{\hyperref[#1]{{Theorem~\ref*{#1}}}}
\newcommand{\lemref}[1]{\hyperref[#1]{{Lemma~\ref*{#1}}}}
\newcommand{\corref}[1]{\hyperref[#1]{{Corollary~\ref*{#1}}}}
\newcommand{\defref}[1]{\hyperref[#1]{{Definition~\ref*{#1}}}}
\newcommand{\secref}[1]{\hyperref[#1]{{Section~\ref*{#1}}}}
\newcommand{\figref}[1]{\hyperref[#1]{{Figure~\ref*{#1}}}}
\newcommand{\tabref}[1]{\hyperref[#1]{{Table~\ref*{#1}}}}
\newcommand{\remref}[1]{\hyperref[#1]{{Remark~\ref*{#1}}}}
\newcommand{\appref}[1]{\hyperref[#1]{{Appendix~\ref*{#1}}}}
\newcommand{\claimref}[1]{\hyperref[#1]{{Claim~\ref*{#1}}}}
\newcommand{\propref}[1]{\hyperref[#1]{{Proposition~\ref*{#1}}}}
\newcommand{\exampleref}[1]{\hyperref[#1]{{Example~\ref*{#1}}}}
\newcommand{\conjref}[1]{\hyperref[#1]{{Conjecture~\ref*{#1}}}}
\newcommand{\factref}[1]{\hyperref[#1]{{Fact~\ref*{#1}}}}

\usepackage{graphicx}
\usepackage{tikz}
\usepackage{subfig}
\usepackage{array}
\usepackage{xspace}

\input{Qcircuit}

\newcommand{\OR}{\textsc{or}\xspace}
\newcommand{\AND}{\textsc{and}\xspace}
\newcommand{\NOT}{\textsc{not}\xspace}

\newcommand{\PARITY}{\textsc{parity}\xspace}
\newcommand{\CNOT}{\textsc{cnot}\xspace}

\newcommand{\Tr}{\operatorname{Tr}}

\renewcommand{\>}{\rangle}
\newcommand{\<}{\langle}

\newcommand{\cl}[1]{\ensuremath{\mathrm{#1}}\xspace}

\begin{document}


\title{Dequantizing read-once quantum formulas}

\author{
\normalsize Alessandro Cosentino, Robin Kothari, and Adam Paetznick \\[.5ex]
\small David R.\ Cheriton School of Computer Science \\
\small and Institute for Quantum Computing \\
\small University of Waterloo
}

\date{}
\maketitle


\maketitle

\begin{abstract}
Quantum formulas, defined by Yao [FOCS '93], are the quantum analogs of classical formulas, i.e., classical circuits in which all gates have fanout one.
We show that any read-once quantum formula over a gate set that contains all single-qubit gates is equivalent to a read-once classical formula of the same size and depth over an analogous classical gate set.
For example, any read-once quantum formula over Toffoli and single-qubit gates is equivalent to a read-once classical formula over Toffoli and \NOT gates. 
We then show that the equivalence does not hold if the read-once restriction is removed. To show the power of quantum formulas without the read-once restriction, we define a new model of computation called the one-qubit model and show that it can compute all boolean functions.  This model may also be of independent interest.
\end{abstract}


\section{Introduction}
\label{sec:intro}

It is widely believed that quantum computers can outperform classical computers for certain problems. Two prominent examples of such problems are factoring, solved by Shor's algorithm~\cite{Sho97}, and simulation of quantum systems~\cite{Fey82,Llo96,AT03}.
Many restricted versions of quantum computers also outperform classical models. Studying the power of restricted quantum models can help identify the ``quantum'' features that are required for computational speedups. 

Some restricted models are also practically motivated, and could be available before we are able to build unrestricted quantum computers.
The ``one clean qubit'' model~\cite{KL98}, for example, can solve some problems that are not known to have efficient classical algorithms~\cite{SJ08}. 
Similarly, log-depth quantum circuits can implement Shor's algorithm with the aid of a classical computer~\cite{CW00}.

On the other hand, if enough restrictions are placed on a quantum model, it may be efficiently simulable by a classical model. In analogy with derandomization, we call this \emph{dequantization}. For example, the simulation of a polynomial-time quantum computer by an exponential-time classical computer is an example of dequantization, albeit a very weak one. On the other hand, if the quantum model is \emph{equivalent} to a classical model, we call this \emph{strong} dequantization. For example, if it were shown that a polynomial-time quantum computer can be simulated by a polynomial-time classical computer, this would be an example of strong dequantization. In this paper, we strongly dequantize read-once quantum formulas by showing that they are equivalent to classical read-once formulas. 

Several past results can be viewed as dequantizing quantum models of computation. For example, Valiant introduced and dequantized a restricted model of quantum computation~\cite{Val02} that was later shown to be equivalent to a classical model \cite{VdNes11}. Additionally, quantum circuits containing only Clifford gates are equivalent to classical circuits of \CNOT and \NOT gates \cite{Got98, BCL+06}. Other examples of dequantization appear in~\cite{FGHZ05, TD04}.

Classical formulas are a well-studied restriction of circuits in which gates have a single output wire and each gate is a function from $k\geq 1$ bits to one bit. Compared to general circuits, the power of formulas is much better understood and several lower bound techniques are known for explicit functions~\cite{Weg87,Juk12}. The study of formulas has lead to a better understanding of the difficulty of proving lower bounds for general circuits. Similarly, studying quantum formulas may lead to a better understanding of the power of quantum circuits.

Indeed, little is known about quantum formulas.  They were defined and examined by Yao~\cite{Yao93} in 1993. But, other than additional results by Roychowdhury and Vatan~\cite{RV02}, they are largely unstudied.

In this paper, we ask whether it is possible to dequantize quantum formulas. Informally, the question is whether, for a quantum gate set $G$, there exists a classical gate set $\hat G$ of roughly equivalent power, such that any quantum formula over $G$ can be written as a classical formula over $\hat G$. We discuss this question in \secref{sec:read-once-formulas} and define an appropriate classical gate set $\hat G$ (\defref{def:classical-formula-gate-set}) corresponding to any quantum gate set $G$. Our main results (\thmref{thm:exact-qrof} and \thmref{thm:bdd-error-qrof}) essentially resolve this question for read-once quantum formulas by showing that read-once quantum formulas over a gate set that includes all single-qubit gates can always be dequantized.

One utility of our classical gate set $\hat G$ is that, in some cases of interest, it corresponds to the gate set that one would naturally expect.
For the set of all $k$-qubit channels (for some constant $k$), which is the gate set used in previous papers on quantum formulas~\cite{Yao93,RV02}, $\hat G$ is the set of all $k$-bit gates. 
The set of arbitrary fanin Toffoli gates and all single-qubit gates is a commonly used gate set in quantum circuit complexity (see, e.g.,~\cite{GHMP02}). We show in \thmref{thm:toffoli-qformula-equals-formula} that, for this gate set, $\hat G$ is the set of classical arbitrary fanin Toffoli gates and the \NOT gate.

It is natural to ask whether the read-once constraint is required for dequantization to hold. In \secref{sec:one-qubit-model} we show that \thmref{thm:exact-qrof} and \thmref{thm:bdd-error-qrof} are false if the read-once constraint is dropped. In particular, we show that there exist quantum formulas over a gate set $G$ that cannot be simulated by classical formulas over $\hat G$. 
To show this we define a model of computation that we call the one-qubit model. Our model is similar to the ``one clean qubit'' model of Knill and Laflamme~\cite{KL98}, but we do not have any mixed states in addition to the one clean qubit. We show that  this model can compute any boolean function (\thmref{thm:oqq-contains-nc1}). However, the one-qubit model is contained in quantum formulas over a gate set $G$ for which $\hat G$ contains only the \NOT gate and the \PARITY gate. Even classical \emph{circuits} of arbitrary size over \NOT and \PARITY cannot compute, for example, the \AND function on two bits.  Thus our dequantization theorems are false without the read-once constraint.

\section{Preliminaries}
\label{sec:prelims}

We start with an introduction to classical formulas and then extend the definition to quantum formulas. We refer the reader to textbooks on circuit complexity~\cite{Weg87,Juk12} or quantum computing~\cite{NC00} for further information.

A gate set $F$ is a set of functions from $k \geq 1$ bits to one bit. A classical formula over a gate set $F$ is a circuit composed of gates from $F$ in which the output of each gate is connected to the input of at most one other gate---i.e., gates have fanout one. Note, however, that input bits may have arbitrary fanout, i.e., more than one gate can use the same bit $x_i$ as an input. The formula outputs a single bit and is said to exactly compute a boolean function $f$ if the output of the formula is $f(x)$ on input $x$.  The restriction that each gate has fanout one makes the circuit look like a tree in which the output gate is the root, non-output gates are internal nodes, and leaves are labeled by input bits. A \emph{read-once} formula is one in which each input bit appears at most once.

Yao defined a quantum formula as a single-output quantum circuit composed of unitary gates in which the path connecting any input to the output is unique. Equivalently, a quantum circuit is a quantum formula if every gate has at most one output that is used as an input to a subsequent gate. The other outputs of a gate are never used again and can be discarded (traced out). We can regard the unitary and discard step as one operation and call that a quantum gate. This makes the analogy with classical formulas clearer. In this paper we use the phrases ``quantum channel'' and ``quantum gate'' interchangeably; a quantum gate need not be unitary. 
We will sometimes talk about a formula over a set of unitaries, which may have multiple output qubits.  In this case, the formula may use any single-output channel obtained by applying one of the unitaries in the set and then tracing out all but one of the output qubits.

We define a quantum gate set to be a set $G$ of quantum channels from $k \geq 1$ qubits to one qubit. A quantum channel is a completely-positive trace-preserving map. In the case that $k=1$ we call the channel a \emph{single-qubit} channel, otherwise we call the channel a \emph{$k$-qubit} channel. Note that a quantum formula may read input bits multiple times, just like a classical formula. A read-once quantum formula is one in which each input bit appears at most once. \figref{fig:circuit-vs-formula} shows an example of a classical formula, a read-once classical formula and a read-once quantum formula. Similar to the classical case, one gate is designated as the output gate. A quantum formula is said to exactly compute a boolean function $f$ if the output of the formula is $|f(x)\>$ on input $x$. In \secref{sec:dequantization-bounded-error} we discuss how to extend this definition to bounded-error quantum formulas. 

\begin{figure}
\centering
\subfloat[]{\includegraphics[height=2.4cm]{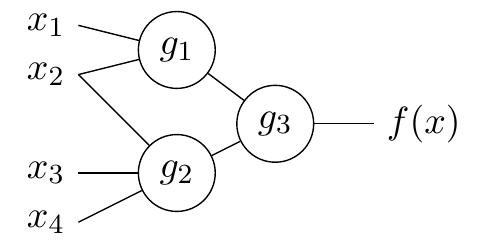}}
\hfill
\subfloat[]{\includegraphics[height=2.4cm]{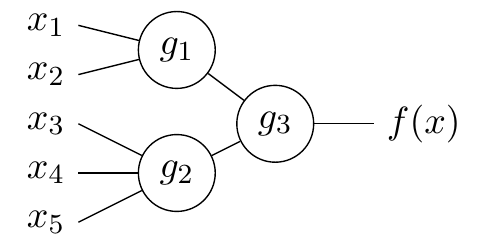}}
\hfill
\subfloat[]{\includegraphics[height=2.4cm]{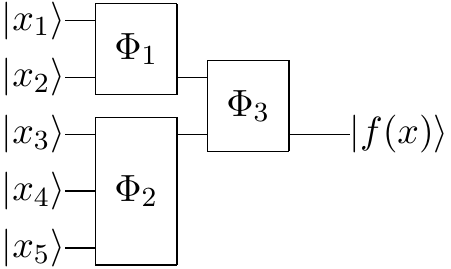}}
\caption{\label{fig:circuit-vs-formula}
(a) A classical (read-many) formula may use input bits multiple times to compute a function $f(x)$.  Here the input $x_2$ is used twice and the formula is over the gate set $\hat G = \{g_1, g_2, g_3\}$. (b) A classical \emph{read-once} formula may use each input bit only once. (c) A quantum read-once formula has the same structure as a classical read-once formula, but may contain quantum channels with single qubit outputs.  Here the quantum formula is over the gate set $G = \{\Phi_1, \Phi_2, \Phi_3\}$.}
\end{figure}

The \emph{size} of a formula is defined as the number of gates, excluding single-bit or single-qubit gates, following standard convention in classical circuit complexity. 
The \emph{depth} of a formula is the maximum number of multi-bit or multi-qubit gates on any path from the output to an input. We say that a formula accepts a language $L \subseteq \{0,1\}^*$ if on input $x$ it outputs $1$ if and only if $x \in L$.

We now state previously known results about quantum formulas. Yao first showed that the majority function needs super-linear sized bounded-error quantum formulas over the set of all three-qubit unitaries~\cite{Yao93}. Later, Roychowdhury and Vatan~\cite{RV02} showed that a classical formula-size lower bound technique due to Nechiporuk~\cite{Weg87,Juk12} also extends to bounded-error quantum formulas over the set of all $k$-qubit unitary matrices, for any constant $k$.

Roychowdhury and Vatan also showed that any function computed by a bounded-error quantum formula over $k$-qubit unitaries can be computed by a classical \emph{circuit} of slightly larger size and depth. Our result applies to a wider range of gate sets and dequantizes to the more limited model of classical read-once formulas.

\section{Dequantization of read-once formulas}
\label{sec:read-once-formulas}

Our main objective is to determine the conditions under which it is possible to dequantize quantum formulas. More precisely, for a quantum gate set $G$ we seek a classical gate set $\hat G$ for which a language $L$ is accepted by a quantum formula over $G$ if and only if it is accepted by a classical formula over $\hat G$.

Note that we want the two classes to be exactly equal in power, thus proving a strong dequantization result. If we only require that all functions computed by quantum formulas over $G$ can be computed by classical formulas, then the result is trivial since we could just allow $\hat G$ to be the set of all boolean functions.

In this section, we show that read-once quantum formulas, for gate sets which include all single-qubit gates, can be strongly dequantized to classical formulas of the same size and depth. 
We then explicitly construct the classical gates sets corresponding to some particular quantum gate sets of interest.
For example, any read-once quantum formula over all $k$-qubit channels is equivalent to a classical read-once formula over all $k$-bit functions.
Similarly, read-once quantum formulas over Toffoli gates and all single-qubit gates can be dequantized to classical read-once formulas over Toffoli and \NOT.

\subsection{Dequantization of exact read-once formulas}
\label{sec:dequantization-exact-case}
We first prove the claim for exact read-once formulas. The proof is simpler in the exact case, and contains all of the essential ideas.  Extension to bounded-error is discussed in~\secref{sec:dequantization-bounded-error}.

Before attempting to prove the theorem, let us discuss the correspondence between the quantum gate set $G$ and its classical counterpart $\hat G$. 
Call a quantum channel \emph{classical} if the output is classical whenever the input is classical. If $G$ contains a classical channel, it is clear that $\hat G$ should contain a gate that performs the same classical operation.

Some gates in $G$ may be non-classical, but may be composed with other gates in a quantum formula to form a classical gate.  Consider, for example, the depth-one quantum formula in~\figref{fig:depth-one-qformula}.  Given a classical input, this formula outputs a classical bit and therefore computes a classical function $f$.  For strong dequantization, we require that $\hat G$ admit a depth-one classical formula that computes $f$.  

\begin{figure}[h]
\centering
\includegraphics[height=2.1cm]{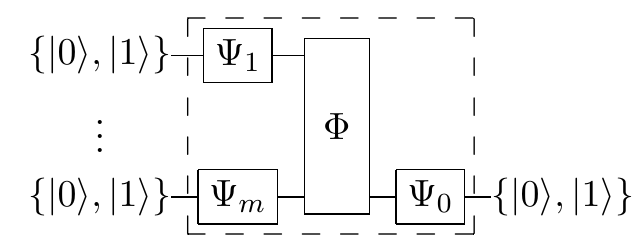}
\caption{\label{fig:depth-one-qformula} A depth-one quantum formula computes a classical function if on classical inputs it outputs a classical bit.}
\end{figure}

We therefore define $\hat G$ to contain only those gates that can be obtained by depth-one quantum formulas over $G$. Informally, this is the smallest gate set that would suffice to prove a strong dequantization result. 

\begin{definition}
\label{def:classical-formula-gate-set}
Let $G$ be a set of quantum channels. Define $\hat G$ to be the set of all classical gates that are computable by a \emph{depth-one} quantum formula over $G$.
\end{definition}

We will also need the following fact about density matrices and quantum channels.

\begin{fact}
\label{fact:output-orthogonal-implies-input-orthogonal}
Let $\Phi$ be a single-qubit channel and $\rho,\sigma$ be $2 \times 2$ density matrices. If $||\Phi(\rho) - \Phi(\sigma)||_1=2$ then $\rho$ and $\sigma$ are pure and orthogonal.
\end{fact}
\begin{proof}
We use the fact that the trace distance is monotone non-increasing under the action of quantum channels~\cite[(9.72)]{Wil11} to conclude that $||\rho - \sigma||_1 \geq 2$. Thus $||\rho - \sigma||_1 = 2$, since the trace distance of two qubits cannot be larger than two~\cite[(9.11)]{Wil11}. Furthermore, the trace distance is maximum only when $\rho$ and $\sigma$ have support on orthogonal subspaces~\cite[(9.12)]{Wil11}. Since these are $2 \times 2$ matrices which act on a two-dimensional vector space, their supports can be at most one-dimensional if the supports are orthogonal. Thus they can be written in the form $\rho=|\psi\>\<\psi|$, $\sigma=|\phi\>\<\phi|$. Finally, since they have support on orthogonal subspaces, $\<\psi|\phi\> = 0$, as claimed.
\end{proof}

We are now ready to state the main theorem.

\begin{theorem}
\label{thm:exact-qrof}
Let $G$ be a set of quantum channels that includes all single-qubit channels. Then a language $L$ is accepted by an exact read-once quantum formula of depth $d$ and size $s$ over $G$ if and only if $L$ is accepted by a read-once classical formula of depth $d$ and size $s$ over $\hat G$, where $\hat G$ is given by~\defref{def:classical-formula-gate-set}.
\end{theorem}

\begin{proof}
One direction of the ``if and only if'' is obvious; we only prove the other direction. The proof is by induction on the depth of the quantum formula. 
The claim is clearly true for depth-one since $\hat G$ contains all classical functions  computable by depth-one quantum formulas over $G$. 
\vspace{-7pt}

\noindent
\begin{minipage}[t]{\textwidth}
\begin{wrapfigure}{r}{0.3\textwidth}
\vspace{-18pt}
\centering
\includegraphics[width=.3\textwidth]{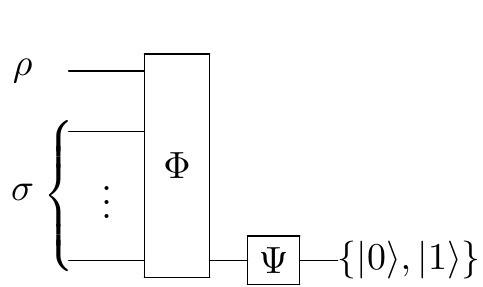}
\caption{\label{fig:final-gate} The final gates $\Phi$ and $\Psi$ of a quantum formula.}
\end{wrapfigure}

\setlength{\parindent}{1.5em}
Consider the final multi-qubit gate $\Phi \in G$ of a depth-$d$ quantum formula that accepts language $L$.
$\Phi$ may be followed by a single-qubit channel $\Psi$, the output of which is a classical state. Let $m$ be the number of inputs to $\Phi$. 

Each input qubit to $\Phi$ is the output of a quantum sub-formula of depth at most $d-1$ on a distinct subset of the original input bits.  We cannot invoke the induction hypothesis directly on the depth-($d-1$) sub-formulas, however, because the outputs may be quantum states.  We will show that each sub-formula can be replaced with different sub-formula of the same size and depth that outputs a classical bit.
\end{minipage}
\vspace{-3pt}

Without loss of generality, let the first input of $\Phi$ depend on the input bits $\{x_1, \ldots, x_l\}$ for some $1 \leq l \leq n$. Let the state of the first input qubit, which is a function of $\{x_1, \ldots, x_l\}$, be called $\rho$. Let $\sigma$ be the state of the remaining input qubits. Thus $\sigma$ is a function of the remaining inputs $\{x_{l+1}, \ldots, x_n\}$.  

Let $P$ be the set of all distinct states $\rho$ obtained by enumerating over all inputs  $\{x_1, \ldots, x_l\}$.  Our goal is to show that either $P$ contains exactly two orthogonal pure states or the output $\Phi$ is independent of the first input qubit. If $P$ contains exactly two orthogonal pure states, then there is a unitary channel $\Psi_1$ such that $\Psi_1(\rho) \in \{\ket 0 \bra 0, \ket 1 \bra 1\}$ for all $\rho \in P$. By composing $\Psi_1$ and $\Phi$ we may construct a new channel that accepts a classical bit as its first input. On the other hand, if the output is independent of the first input, we can consider a channel $\Psi_1$ that always outputs a fixed state independent of its input. 

First, we look for two density matrices in $P$, which we will call $\rho_1$ and $\rho_2$, and a bit string $x_{l+1}\ldots x_n$ such that the induced state $\sigma$ satisfies $\Phi(\rho_1 \otimes \sigma) \neq \Phi(\rho_2 \otimes \sigma)$. Since the formula is read-once, bits $\{x_1, \ldots, x_l\}$ are independent of bits $\{x_{l+1}, \ldots, x_n\}$. Thus if there are no $\rho_1$, $\rho_2$ and $\sigma$ that satisfy the above condition, then the output of $\Phi$ is independent of $\{x_1, \ldots, x_l\}$ and we may replace $\rho$ by a channel $\Psi_1$ that on any input (and classical inputs in particular) outputs some fixed element of $P$.
Note that this is the only place in the proof in which the read-once condition is required.

Assume that such a $\rho_1$, $\rho_2$ and $\sigma$ exist. For this fixed $\sigma$, define $\Phi'(\rho) := \Psi(\Phi(\rho \otimes \sigma))$.  The output of $\Phi'$ is a classical state and therefore $\rho_1$ and $\rho_2$ induce orthogonal classical outputs. Thus $||\Phi'(\rho_1) - \Phi'(\rho_2)||_1=2$ and   by \factref{fact:output-orthogonal-implies-input-orthogonal}, we know that
$\rho_1$ and $\rho_2$ are pure and orthogonal. Let $\Psi_1$ be a unitary channel such that $\Psi_1(\rho_1) = \ket 0 \bra 0$ and $\Psi_1(\rho_2) = \ket 1 \bra 1$.

We now show that $|P| = 2$.  That is, $\rho_1$ and $\rho_2$ are the \emph{only} states in $P$.  Assume that $|P| > 2$, and let $\rho_3 \in P$ be distinct from $\rho_1$, $\rho_2$ and such that $\Phi'(\rho_3) \neq \Phi'(\rho_1)$.  Since $\Phi'(\rho)$ is classical for any $\rho \in P$, using \factref{fact:output-orthogonal-implies-input-orthogonal} we again have that $\rho_1$ and $\rho_3$ are pure and orthogonal. But since $\rho_2$ is the unique state orthogonal to $\rho_1$, we conclude that $\rho_3 = \rho_2$, a contradiction.  Assuming that $\Phi'(\rho_3) \neq \Phi'(\rho_2)$ similarly leads to a contradiction.

A similar argument applies for the set of possible states on the remaining input qubits of $\Phi$.  For each qubit $k$ there are two possible (pure) states $\rho_0$ and $\rho_1$ and a unitary channel $\Psi_k$ such that $\Psi_k(\rho_0) = \ket{0}\bra{0}$ and $\Psi_k(\rho_1) = \ket 1 \bra 1$, or the action of $\Phi$ is independent of qubit $k$ and we may replace it with an input-independent channel $\Psi_k$. We now add the gates $\Psi_k$ before gate $\Psi$ on input qubit $k$. If $\Psi_k$ was an input-independent channel, we have not changed the output of the circuit. If $\Psi_k$ was a unitary channel that changes basis, we now need to add $\Psi^\dagger_k$ before it to ensure that the output is unchanged. The channel formed by the $\Psi_k$s, $\Phi$ and the output gate $\Psi'_m$ is now classical and has the same action as a classical gate $R \in \hat G$ as shown in~\figref{fig:quantum-to-classical-channel-replacement}. 

\begin{figure}
\centering
\includegraphics[height=3.2cm]{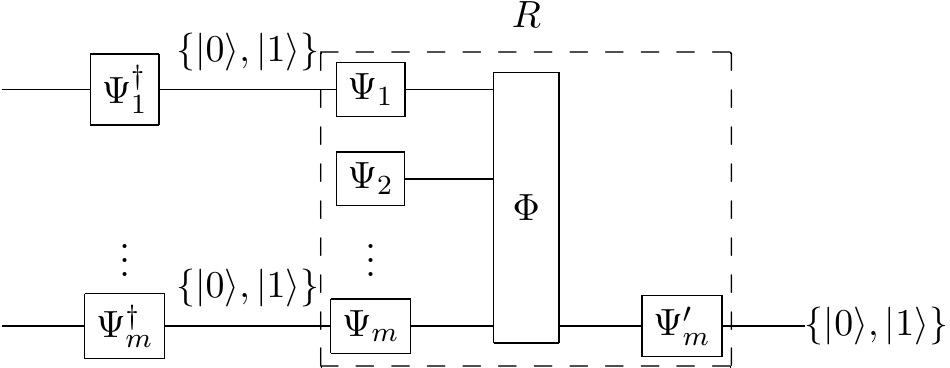}
\caption{\label{fig:quantum-to-classical-channel-replacement}
Single qubit channels are prepended to input qubits.  If the input qubit depends on the input to the formula then both the channel and its adjoint are prepended. Otherwise, the input qubit is fixed and is replaced by a single channel. The channel $\Phi$ is then replaced by a classical gate $R \in \hat G$.
}
\end{figure}

The inputs to $R$ are quantum sub-formulas of depth at most $d-1$, each of which outputs a classical bit.  By induction, each sub-formula may be replaced by a classical formula of the same depth over $\hat G$.  The resulting circuit is a depth-$d$ classical formula over $\hat G$ that accepts language $L$.
\end{proof}

\subsection{Extension to the bounded-error case}
\label{sec:dequantization-bounded-error}

We now extend the main result to the bounded-error setting. 
The proof is essentially the same as that of~\thmref{thm:exact-qrof}, but does not require that the formula outputs be orthogonal states.
This section can be safely skipped without loss on continuity.

There are several natural definitions of bounded-error quantum formulas. For example, we could say that a quantum formula over a gate set $G$ computes a function $f$ with error $\epsilon$, if on input $x$, the output of the formula is  $\epsilon$-close to $\ket{f(x)}$ in trace distance. However, for bounded-error, most reasonable definitions will be equivalent (up to constant factors). So we use the following definition which is more convenient for our proofs: A quantum formula over $G$ computes $f$ with error $\delta$ if all the output states corresponding to $f(x)=0$ are at least $2-\delta$ away in trace distance from all the output states corresponding to $f(x)=1$. 

This definition is equivalent to the previous one, up to constants, by noting that if one state is close to $\ket{0}$ and the other is close to $\ket{1}$, then they must necessarily by far apart, which can be proved using the triangle inequality. In the other direction, using a result of Gutoski and Watrous~\cite{GW05}, it can be shown that two sets of states with a lower bound on the minimum pairwise distance can be distinguished with high probability.  Note that when $\delta = \epsilon = 0$, the definitions coincide, except that the second definition only requires the output states to be orthogonal, not necessarily $\ket{0}$ and $\ket{1}$. Since our gate 
sets contain all single-qubit gates, this distinction is not important. The definition of $\hat{G}$ will also have to be similarly modified to incorporate bounded-error quantum formulas over $G$.

In the exact case, the output of a formula is always either $\ket{0}$ or $\ket{1}$. Furthermore, we showed that even the output of sub-formulas is classical, up to a change of basis. We now want to prove a similar claim for bounded-error sub-formulas. We wish to show that if the function depends on the output of a sub-formula, then the set of output states of the sub-formula can be partitioned into two non-empty sets, $S_0$ and $S_1$, such that the trace distance between any pair of states in $S_0$ and $S_1$ is at least $2-\delta$.

Now we can state the bounded-error analog of \thmref{thm:exact-qrof}. Note that the definition of $\hat{G}$ used in this theorem is different from that in \thmref{thm:exact-qrof}, since it allows functions to be computed with bounded error.

\begin{theorem}
\label{thm:bdd-error-qrof}
Let $G$ be a set of quantum channels that includes all single-qubit channels. If a language is accepted by a bounded-error read-once quantum formula over $G$ then it is also accepted by an exact read-once classical formula over $\hat G$, of the same size and depth, where $\hat G$ is defined as the set of all classical gates that can be computed by a depth-one bounded-error read-once quantum formula.
\end{theorem}

\begin{proof}
The proof proceeds like the proof of \thmref{thm:exact-qrof}. We use the same notation as in the other proof. As before, we consider the first input of the last multi-qubit gate, and assume it depends on the input bits $\{x_1, x_2, \ldots, x_l\}$, and let $P$ denote the set of all states obtained by enumerating over the the input bits  $\{x_1, x_2, \ldots, x_l\}$. Denote the state on the rest of the inputs to the last gate as $\sigma$.

First we look for two states in $P$, $\rho$ and $\rho'$, and a bit string $x_{l+1}\ldots x_n$ which induces a state $\sigma$ on the other inputs such that $\rho \otimes \sigma$ and $\rho' \otimes \sigma$ lead to different outputs of the formula. If no such $\rho$, $\rho'$ and $\sigma$ exist, then the output of the last gate is independent of $\{x_1, \ldots, x_l\}$ and we may replace the first input qubit by a channel that on any input (and classical inputs in particular) outputs some fixed element of $P$.

Assume that such a $\rho$, $\rho'$ and $\sigma$ exist. As before, for this fixed $\sigma$, define $\Phi'(\rho) := \Psi(\Phi(\rho \otimes \sigma))$. 
Without loss of generality, assume that $\rho$ leads to output $0$ and $\rho'$ leads to output 1. Let $S_b :=\{\rho \in P: \rho \text{ leads to output $b$}\}$. Then for any states $\rho_0 \in S_0$ and $\rho_1 \in S_1$ we must have  $||\Phi'(\rho_0) - \Phi'(\rho_1)||_1 \geq 2-\delta$. Using the monotonicity of trace distance under quantum channels~\cite[(9.72)]{Wil11}, this implies $||\rho_0 - \rho_1||_1 \geq 2-\delta$, which means all the states in $S_0$ are far from all the states in $S_1$. Thus the first sub-formula satisfies the definition of a bounded-error quantum formula. By the induction hypothesis, there is a classical formula over $\hat{G}$ of the same size and depth that outputs $0$ when the quantum sub-formula would have output some state in $S_0$ and outputs $1$ when it would have output some state in $S_1$.

Now we wish to show that all the states in $S_b$ are equivalent from the perspective of the last gate. More precisely, let $\rho_0$ and $\rho'_0$ be two different states in $S_0$, and let $\sigma$ be some valid state induced by $x_{l+1}, \ldots x_n$. Let us also partition the output states of the last gate into two sets  $T_0$ and $T_1$, for which the trace distance between any pair of states in the two sets is at least $2 - \delta$. We wish to show that if the output corresponding to $\rho_0 \otimes \sigma$ is in $T_b$, then so is the output corresponding to $\rho'_0 \otimes \sigma$. 
States $\rho_0 \otimes \sigma$ and $\rho'_0 \otimes \sigma$ arise due to valid input strings $x$ and $x'$, respectively. Thus the corresponding outputs must be in $T_0$ or $T_1$. Since the output of $\rho_0 \otimes \sigma$ is in $T_b$, and since $\rho_0$ cannot be too far from $\rho'_0$,  the output of $\rho'_0 \otimes \sigma$ is also in $T_b$. More precisely, since both $\rho_0$ and $\rho'_0$ are at least $2-\delta$ away from some state $\rho_1 \in S_1$, we know that they can be at most $2\delta$ distance apart. Thus, by monotonicity of trace distance under quantum channels, their corresponding outputs cannot be too far apart.

Now we know that all the states in $S_b$ are equivalent from the perspective of the last gate. We also know that there is a classical formula over $\hat{G}$ of the same size and depth that outputs $0$ when the quantum sub-formula would have output some state in $S_0$ and outputs $1$ when it would have output some state in $S_1$.
We can just add a single-qubit gate to the output of this sub-formula, which on input $\ket{0}$ outputs some fixed state $\rho_0$ from $S_0$ and on input $\ket{1}$ outputs a fixed state $\rho_1$ from $S_0$. This formula also computes the same function $f(x)$, as shown in the previous paragraph. 

Continuing this on all the inputs to the last gate, we obtain a depth-one quantum formula which accepts as inputs the outputs of a classical depth $d-1$ formula. Using the definition of $\hat{G}$, there is a classical gate that computes the same function, which gives us a classical depth-$d$ formula for the entire function.
\end{proof}

\subsection{Application to concrete gate sets}
\label{sec:concrete-gate-sets}
A simple corollary of~\thmref{thm:exact-qrof} is that $L$ is accepted by a read-once quantum formula over the set of all $k$-qubit channels (for some constant $k$) if and only if it is accepted by a read-once classical formula over the set of all $k$-bit functions. This is the gate set used in the previous studies of quantum formulas~\cite{Yao93,RV02}.

Another gate set of interest is the set of arbitrary fanin Toffoli gates and all single-qubit gates.  This gate set is commonly used in the study of quantum circuits~\cite{GHMP02}. We now explicitly construct $\hat G$ for this gate set and show that it only contains classical arbitrary fanin Toffoli gates and the \NOT gate.

Let us define $\Phi^{\text{Tof}}_m$ to be the quantum channel obtained from a Toffoli gate with $m \geq 2$ qubits by tracing out all the $m-1$ control qubits. We can assume that all Toffoli gates that appear in the formula always output the target qubit and trace out the control qubits because by conjugating the target qubit and a control qubit with the Hadamard matrix, $H$, it is possible to exchange the roles of the target and that control qubit.

To compute $\hat G$ for the set of all single-qubit gates and $\Phi^{\text{Tof}}_m$ for $m \geq 2$ it suffices to list all possible classical gates that can be obtained from $\Phi^{\text{Tof}}_m$ by placing single-qubit channels before and after it. The classical $m$-bit Toffoli gate computes the function $f(x_1, x_2, \ldots, x_m) = (x_1 \wedge x_2 \wedge \ldots \wedge x_{m-1}) \oplus x_m$. Let $F^{\text{Tof}}_m$ be the set of all functions obtained from the classical $m$-bit Toffoli by placing single-bit gates before and after it. The only single-bit gates are \NOT gates, and channels that output a constant bit.

\begin{theorem}
\label{thm:toffoli-qformula-equals-formula}
Let $f$ be any function on $x \in \{0,1\}^m$ that can be obtained by placing single-qubit gates before and after $\Phi^{\textrm{Tof}}_m$. Then $f \in F^{\textrm{Tof}}_m$.
\end{theorem}

\begin{proof}
Let the classical function $f$ be defined by single-qubit channels $\Psi_1, \ldots, \Psi_m$ on the inputs followed by $\Phi^{\textrm{Tof}}_m$ and a single-qubit quantum channel $\Psi'_m$ on the output qubit. 

Let $\{\rho_1, \ldots, \rho_m\}$ be the inputs to $\Phi^{\textrm{Tof}}_m$ (i.e., the outputs of $\Psi_1, \ldots, \Psi_m$) induced by a particular choice of $x$, and $\rho'_m$ be the output of the channel (i.e., the input to $\Psi'_m$). See~\figref{fig:depth-one-qformula-with-toffoli}.

\begin{figure}
\centering
\includegraphics[height=4cm]{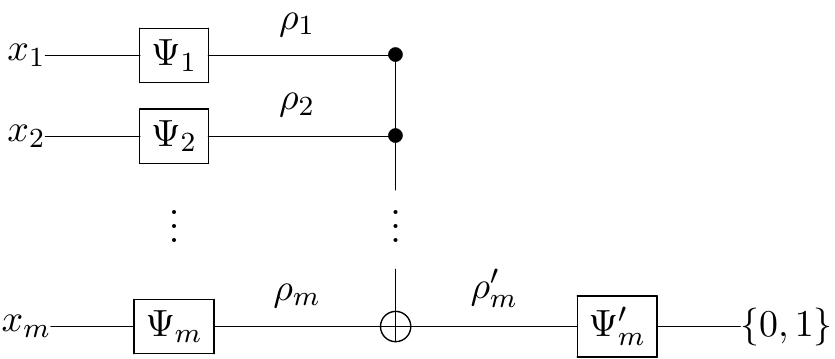}
\caption{\label{fig:depth-one-qformula-with-toffoli}
A depth-one quantum formula over single-qubit gates and an $m$-bit Toffoli gate.
}
\end{figure}

First let us handle the trivial cases. If $\Psi_m$ is a channel which always outputs an eigenstate of $X$, i.e., $X \rho_m X = \rho_m$ for $x_m = 0$ and $x_m = 1$, then the Toffoli gate leaves it unaffected. Thus the output  $\rho'_m$ is only a function of $x_m$, and since the output of the gate is classical, it is some classical one-bit function of $x_m$, all of which are contained in $F^{\textrm{Tof}}_m$. We may now assume this is not the case.
Let us also assume that $f$ outputs $0$ on some input and $1$ on some input, otherwise it is a constant function.

For any $x, x' \in \{0,1\}^{m-1}$, define $\alpha(x,x') = \prod_{i=1}^{m-1} \rho_i(x_i,x'_i)$. Furthermore, let $0 \leq a \leq 1$ be defined as $a = \alpha(11\ldots 1, 11 \ldots 1)$.
Then the input to $\Phi^{\textrm{Tof}}_m$ is $\rho_1 \otimes \rho_2 \otimes \ldots \otimes \rho_m = \sum_{x,x'} \alpha_{xx'} \ket{x}\bra{x'} \otimes  \rho_m$. The output of $\Phi^{\textrm{Tof}}_m$ is given by
\begin{equation}
\begin{split}
  \rho'_m
  &= \Tr_{1,\ldots,m-1} \left ( \sum_{x,x'} \alpha_{xx'} \ket{x}\bra{x'} \otimes X^{\AND(x)} \rho_m X^{\AND(x')} \right) \\
  &= \sum_{x} \alpha_{xx} X^{\AND(x)} \rho_m X^{\AND(x)} \\
  &= \left [\prod_{i=1}^{m-1} \rho_i(1,1) \right] X \rho_m X + \left[ \sum_{x \neq 1\ldots1} \prod_{i=1}^{m-1} \rho_i(x_i,x_i) \right] \rho_m \\
  &= a X \rho_m X + (1-a) \rho_m.
\end{split}
\label{eq:toffoli-output-qubit}
\end{equation}

By assumption, $\Psi'_m(\rho'_m)$ is classical and both outputs $0$ and $1$ are possible, thus by \factref{fact:output-orthogonal-implies-input-orthogonal}, $\rho'_m$ must always be pure. Since we have also assumed that $X \rho_m X \neq \rho_m$, \eqnref{eq:toffoli-output-qubit} implies that $\rho'_m$ can be pure only if $a=0$ or $a=1$. But $a$ takes values $0$ or $1$ on all inputs, so the following short argument implies that each $\rho_i$ must be classical (i.e., $\rho_i(0,0) = 0$ or $\rho_i(0,0) = 1$) for every input $x$. Toward a contradiction, assume this is not the case. Thus, for some $i$, there exists an input $x_i$ so that $\rho_i$ is not classical and, in particular, $\rho_i(1,1) \neq 0$ and $\rho_i(1,1) \neq 1$. On this input $x$, $\prod_{i=1}^{m-1} \rho_i(1,1)$ cannot be $0$ or $1$. But this quantity is $a$, which could only be 0 or 1. Thus we have reached a contradiction, and all the $\rho_i$ must be classical for every input $x_i$. But this means every channel $\Psi_i$ is a 
classical one-bit function of $x_i$. Thus the inputs to $\Phi^{\textrm{Tof}}_m$ are always classical and so $\Phi^{\textrm{Tof}}_m$ can be replaced with its classical equivalent, which completes the proof.
\end{proof}

This theorem shows that the only classical gates obtainable from $\Phi^{\text{Tof}}_m$ by placing single-qubit channels before and after it are the classical Toffoli gate of size $m$ and gates derived from it using single-bit gates. In particular, this means that if the gate set $G$ contains all single-qubit channels and the \CNOT gate, which is the Toffoli gate for $m=2$, then $\hat G$ contains only the \CNOT gate and the \NOT gate. We use this result in~\secref{sec:one-qubit-model}.

\section{One-qubit model}
\label{sec:one-qubit-model}

Informally, \thmref{thm:exact-qrof} and \thmref{thm:bdd-error-qrof} show that read-once quantum formulas are equivalent to classical read-once formulas. We now show that the claim is no longer true if we drop the read-once constraint. To this end, we introduce a new model, which we call the \emph{one-qubit} model.  This model is independent of the previous sections, and may be of interest outside of the context of quantum formulas.

The one-qubit model consists of a single qubit initialized to $\ket{0}$ followed by an alternating sequence of single-qubit unitaries and \CNOT (controlled-$X$) gates. The control of each \CNOT is a bit $x_i$ of the input $x$. The output is determined by a measurement in the standard basis. See \figref{fig:oqq}. We call an algorithm of this kind a one-qubit program. We say that a one-qubit program exactly computes a function $f$ if the output of the measurement is $f(x)$ with probability one on input $x$.

\begin{figure}[h]
\centering
\includegraphics[height=1.2cm]{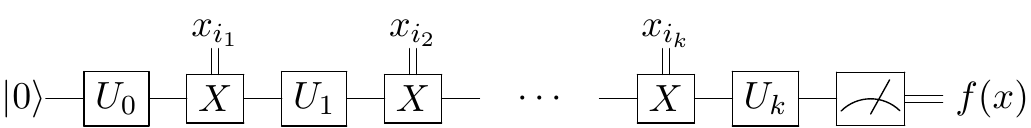}
\caption{\label{fig:oqq}
 A one-qubit program of length $k$.}
\end{figure}

The \emph{length} of a one-qubit program is defined as the total number of \CNOT gates.  For instance, the program in \figref{fig:oqq} has length $k$. Notice that the model is unchanged if, instead of \CNOT, we choose any controlled-$V$ gate such that $UVU^\dagger = X$ for some single-qubit unitary $U$.
Note also that input bits may be re-used as many times as desired.

We prove that the one-qubit model is universal, that is, it can compute all boolean functions.
Specifically, if a boolean function has a depth-$d$ circuit over fanin-$2$ \AND and \OR, and \NOT gates, then it has a one-qubit program of length $4^d$.
Here, the depth of a circuit is defined as the maximum number of \AND or \OR gates on any path from the output to an input.
Note that all boolean functions can be expressed as circuits of polynomial depth, thus one-qubit programs of exponential length can compute all boolean functions.
Moreover, one-qubit programs with polynomial length can exactly compute any function in \cl{NC^1}, the set of functions computed by log-depth poly-size bounded-fanin circuits over \AND, \OR and \NOT. Our proof resembles the original proof of Barrington's theorem~\cite{Bar89}, a surprising result in complexity theory, which showed that bounded-width branching programs can compute any function in \cl{NC^1}.

\begin{theorem}
\label{thm:oqq-contains-nc1}
The one-qubit model can compute all boolean functions. More precisely, any function that has a depth-$d$ circuit over fanin-$2$ \AND and \OR, and \NOT gates can be computed exactly by a one-qubit program of length $4^d$.
\end{theorem}

\begin{proof}
Use the notation $X^{x}$ to denote a \CNOT gate that is controlled by the variable $x$.
Let $C$ be a circuit of depth $d$ that computes a function $C(x)$. Starting from the circuit $C$, we construct a one-qubit program $F$ of length $4^d$ that computes the same function. That is, $F = X^{C(x)}$ so that $F|0\> = |C(x)\>$.

We prove the claim by induction on the structure of the circuit. Circuit $C$ can be seen as a binary tree with the input variables as the leaves of the tree and the gates as the internal nodes. Given a gate at the root of the circuit (\AND or \NOT), 
we assume that the induction hypothesis holds for the subcircuits that produce the inputs of the gate and 
we prove the claim for the entire circuit.

Let us start with the base case. 
If $C$ is a single variable $x_{i}$, the corresponding one-qubit program is a single \CNOT gate controlled by $x_{i}$, i.e., $F = X^{x_{i}}$.

The induction step consists of two cases. Suppose that $C$ is composed of a \NOT gate and a subcircuit $C'$, and let $F'$ be the one-qubit program of length $4^{d}$ that simulates $C'$.
Then we can append a Pauli-$X$ gate to $F'$ and obtain the program $F$ of the same length $4^{d}$ that computes $C$, i.e., $F = XX^{C'(x)} = X^{\bar{C'}(x)} = X^{C(x)}$. Note that the depth of circuits $C$ and $C'$ is the same, since \NOT gates do not contribute to the depth. Also note that single-qubit unitaries do not contribute to the length of the program, and thus the length is unchanged.

Now suppose $C$ is composed of an \AND gate that connects two subcircuits $C'$ and $C''$ and let $F'$ and $F''$ be the programs of length $4^{d-1}$ that compute $C'$ and $C''$, respectively.
Then consider the following program $F$ described as an equation:
\begin{equation}
F = VX^{C'(x)}H^{C''(x)}X^{C'(x)}H^{C''(x)}V = V(iY)^{C'(x) \wedge C''(x)}V = iX^{C(x)}, 
\end{equation}
where $V = \frac{1}{\sqrt{2}}(X + Y)$ is the unitary that satisfies $V^2 = \mathbbm{1}$, $VYV = X$ and $VXV=Y$.

This program effectively computes the \AND of two sub-programs. It starts with a $V$ gate, followed by the subprogram $F'$, which is equivalent to $X^{C'(x)}$. Then we need a subprogram that performs $H^{C''(x)}$. By the induction hypothesis, we have a subprogram $F''$ that performs $X^{C''(x)}$. Since the Hadamard matrix and Pauli X gate are unitarily equivalent, there is a unitary matrix $R$ such that $RXR = H$ and $RHR = X$. Conjugating $X^{C''(x)}$ with $R$ gives us $RX^{C''(x)}R$, which is the same as $(RXR)^{C''(x)}$, which is $H^{C''(x)}$. The other gates in $F$ are constructed similarly.

Thus the program $F$ requires some single qubit gates, two copies of the program for $F'$ and two copies of the program for $F''$. Since $F'$ and $F''$ have length at most $4^{d-1}$ and single qubit unitaries do not count towards the length, we get a program $F$ of length $4^d$, which performs $X^{C(x)}$ up to an irrelevant global phase.

These two cases suffice to prove the theorem since we can replace the \OR gates in the circuit by \AND and \NOT gates, without increasing its size or depth.  
\end{proof}
In the case of \cl{NC^1}, the depth of the circuit is at most $O(\log n)$, therefore the length of the resulting one-qubit program is polynomial.
\begin{corollary}
Any function in \cl{NC^1} can be computed exactly by a one-qubit program of polynomial length.
\end{corollary}

Now that we have shown that the one-qubit model can compute all boolean functions, it remains to show that our dequantization claim is false without the read-once constraint. First observe that the one-qubit model is a restricted model of quantum formulas over the gate set $G$ consisting of all single-qubit gates and the \CNOT gate. 

We know from \thmref{thm:toffoli-qformula-equals-formula} (case $m=2$) that, in this case, the classical gate set $\hat G$ consists of only the \NOT gate and \PARITY gate.
But no formula and, indeed, no \emph{circuit} of any size over $\hat G$ can compute the \AND of two bits, since \NOT and \PARITY do not form a complete basis. In particular, circuits over $\hat G$ can only compute functions expressible as degree-one polynomials over $\mathbb{F}_2$. Thus we conclude that \thmref{thm:exact-qrof} and \thmref{thm:bdd-error-qrof} are false if the read-once constraint is dropped.

Readers familiar with quantum branching programs~\cite{NHK00,SS05,AGK+05} may notice that the one-qubit model is contained in exact width-two quantum branching programs. However, in our model the only input-dependent gate is a Pauli-$X$, whereas in quantum branching programs, we can apply arbitrary input-controlled unitaries. By conjugating the $X$ gate with arbitrary single qubit gates the one-qubit model can obtain any unitary matrix with the same eigenvalues as $X$, but not an arbitrary unitary matrix. Thus it is not clear that results about width-two quantum branching programs can be ported over to our model.

\section{Conclusions and open problems}
\label{sec:conclusions}

We have shown that read-once quantum formulas over any gate set are only as powerful as read-once classical formulas over a related gate set.  As a concrete example, we showed that read-once quantum formulas over Toffoli and single-qubit gates dequantize to the natural analog, read-once classical formulas over Toffoli and \NOT gates.
Perhaps our results may be extended to constant-depth quantum circuit classes, many of which are defined over Toffoli and single-qubit gates, e.g.,~\cite{GHMP02}.  Our proof technique fails when the formula restriction is lifted, but for constant-depth it may be possible to reuse some of the same ideas.

Another obvious open problem is to dequantize all quantum formulas, not just read-once formulas. Although we show that our classical gate set is insufficient to do so, there might be another classical gate set that works. 

Finally, our dequantization result implies that lower bounds on read-once quantum formulas may be obtained from analogous classical lower bounds.
For related models, including general (i.e., not read-once) quantum formulas and constant-depth quantum circuits, the problem of finding non-trivial lower bounds remains. 

\subparagraph*{Acknowledgements}
The authors would like to thank Andrew Childs, Ben Reichardt and John Watrous for insightful comments.

This work was supported by Canada's NSERC, MITACS, the Ontario Ministry of Research and Innovation, the U.S. ARO, and the Mprime Network.

\bibliographystyle{plain}
\bibliography{refs}

\end{document}